\RequirePackage{fix-cm}
\documentclass[smallextended]{svjour3}       
\smartqed  
\usepackage{amssymb,amsmath}
\usepackage[usenames,dvipsnames]{color}
\usepackage[colorlinks=true,linkcolor=red,citecolor=blue]{hyperref}
\usepackage{makecell}
\usepackage{multirow}
\usepackage{eulervm}
\usepackage{graphicx}

\usepackage[usenames,svgnames,dvipsnames]{xcolor}

\newdimen\prevdp
\def\leftlabel#1{\noalign{\prevdp=\prevdepth
   \kern-\prevdp\nointerlineskip\vbox to0pt{\vss\hbox{\ensuremath{#1}}}\kern\prevdp}}

\usepackage{url}

\usepackage{enumerate}

\usepackage{xspace}

\usepackage{nicefrac}

\newcommand{\eps}{\ensuremath{\varepsilon}\xspace}
\renewcommand{\epsilon}{\eps}

\newcommand{\ignore}[1]{}

\renewcommand{\leq}{\leqslant}

\usepackage{graphicx}
\usepackage{natbib}
\newtheorem{mydef}{Definition}
\newtheorem{mytheo}{Theorem}
\newtheorem{mylemma}{Lemma}
\newtheorem{myobservation}{Observation}
\usepackage{tcolorbox}
\usepackage{boldline}
\usepackage{booktabs}
\usepackage{subcaption}
\usepackage[linesnumbered,ruled,vlined]{algorithm2e}
\usepackage[noend]{algpseudocode}
\usepackage[]{algorithm2e}
\usepackage{mathtools}

\DeclarePairedDelimiter\floor{\lfloor}{\rfloor}
\begin{document}

\title{Designing and Connectivity Checking of Implicit Social Networks from the User\mbox{-}Item Rating Data\thanks{The author is supported by the Post Doctoral Fellowship Grant provided by Indian Institute of Technology Gandhinagar (Project No. MIS/IITGN/PD-SCH/201415/006). A small part of this study has been previously published as \cite{banerjee2017algorithms}.}
}


\author{Suman Banerjee        
}


\institute{Suman Banerjee \at
              Department of Computer Science and Engineering, \\
              Indian Institute of Technology, Gandhinagar.\\
              \email{suman.b@iitgn.ac.in}           
}

\date{Received: date / Accepted: date}

\maketitle

\begin{abstract}
\emph{Implicit Social Network} is a connected social structure among a group of persons, where two of them are linked if they have some common interest. One real\mbox{-}life example of such networks is the implicit social network among the customers of an online commercial house, where there exist an edge between two customers if they like similar items. Such networks are often useful for different commercial  applications such as \textit{target advertisement}, \textit{viral marketing}, etc. In this article, we study two fundamental problems in this direction. The first one is that, given the user\mbox{-}item rating data of an E\mbox{-}Commerce house, how we can design implicit social networks among its users and the second one is at the time of designing itself can we obtain the connectivity information among the users. Formally, we call the first problem as the \textsc{Implicit User Network Design} Problem and the second one as \textsc{Implicit User Network Design with Connectivity Checking} Problem. For the first problem, we propose three different algorithms, namely \emph{`Exhaustive Search Approach'}, \emph{`Clique Addition Approach'}, and \textit{`Matrix Multiplication\mbox{-}Based Approach'}. For the second problem, we propose two different approaches. The first one is the sequential approach: designing and then connectivity checking, and the other one is a concurrent approach, which is basically an incremental algorithm that performs  designing and connectivity checking simultaneously. Proposed methodologies have experimented with three publicly available rating network datasets such as \emph{Flixter}, \textit{Movielens}, and \textit{Epinions}. Reported computational time  shows that the `Clique Addition Approach' is the fastest one for designing the implicit social network. For designing and connectivity checking problem the concurrent approach is faster than the other one. We have also investigated the scalability issues of the algorithms by increasing the data size. 
\keywords{Social Networks \and Rating Data \and Clique \and Connected Component}
\end{abstract}

\section{Introduction} \label{intro}
A \emph{social network} is an interconnected structure among a group of agents that is formed for social interactions \cite{wasserman1994social}. Here, agents may be the customers of a commercial house, researchers, etc. and their relationship is  friendship, co\mbox{-}authorship, respectively. These are nowadays open platforms, where information, rumors, ideas, innovations, etc. spread widely and rapidly. Use of social networks varies from the prediction of customer behavior \cite{goel2013predicting} to understanding the \textit{sms wormhole} propagation \cite{xiao2017modeling}. One of the important phenomena of social networks is the \emph{information diffusion} and this means that if a user has some information then he or she tends to share it with his or her neighbors. Thus information propagates from one part of the network to the other. This phenomenon has been exploited by the E\mbox{-}Commerce houses and found potential applications in \emph{viral marketing} \cite{chen2010scalable}, \emph{computational advertisement} \cite{huh2017considerations}, \emph{personalized recommendation} \cite{zhang2013socconnect}, finding influential twitters \cite{riquelme2016measuring}, feed ranking \cite{bonchi2013meme} and so on. Due to different commercial applications of social networks, the last one and half decades have  witnessed a significant interest in mining and analyzing social networks. Look into \cite{aggarwal2014evolutionary} and \cite{al2018analysis} for recent surveys.

\par Based on the design methodology, social networks are of two types: \textit{explicit social networks} (e.g. \emph{Twitter}, \emph{Facebook},  etc.) where users choose their friends by themselves and \textit{implicit social networks} \cite{losup2014analyzing} (e.g. \textit{Epinions}, \textit{Flixter}, etc.) where people are connected based on their common interest; i.e., two users are linked if they have rated (or liked or searched) similar items. For different commercial applications of social networks such as \textit{viral marketing} \cite{domingos2005mining}, \textit{computational advertisement}, \textit{item recommendation} \cite{yang2013bayesian} prior knowledge of the user's on\mbox{-}line behavior is important. For performing these commercial activities sometimes the implicit social network is preferred over explicit one due to the following two reasons. Firstly, in implicit social networks, users are connected based on similar item preferences. Hence, a neighbor's preference can be exploited to predict the preferences of a user with the unknown identity. Secondly, the network is designed and maintained by the E\mbox{-}Commerce house itself, hence it is completely accessible to them \cite{hill2005viral}. So, it is an important question, how to design the implicit social network in a given context. It is surprising to see that the literature in this direction is very limited. To the best of the author's knowledge, other than the \cite{podobnik2015implicit} and \cite{van2014analyzing} there does not exist any study that deals with this problem. From the E\mbox{-}Commerce house perspective, one viable data where interactions between users and items are recorded is the \emph{user\mbox{-}item rating data}. In this paper, we initiate the study of the designing implicit user network \footnote{As, in this study, we are concerned with the designing the implicit social network, where the customers of an E\mbox{-}Commerce house are the users of the network, hence in the rest of the paper we use the two terms: `implicit user network' and `implicit social network' interchangeably.} from rating data.
\par Knowledge regarding the structure of the implicit user network is important for many commercial applications. Think of a situation when an E\mbox{-}Commerce house does viral marketing for its newly launched product. For this purpose, they distribute a number of sample items to influential users with a hope that a significant number of them will be likeing it and start sharing the message among their friends in the network. This diffusion phenomenon will be continued and at the end, majority of the users will come to know about the item. The key issue that comes out in the described context is that which users should be chosen initially for initiating the information diffusion. This problem is popularly known as the influence maximization problem and the users who initiate the diffusion process is called as the `seed users' or `seed nodes' \cite{banerjee2018survey}. Now, it is important to observe that the influence of a seed user can not go beyond the connected component in which it belongs. So, it is important during the seed set selection for the influence maximization process, the component information of the implicit user network should be exploited. Hence, from the described context, not only the designing but also connectivity checking is an important problem. In this paper, along with the designing of implicit user network, we also study the problem of designing as well as connectivity checking of this network. Particularly, we make the following contributions in this paper:
\begin{itemize}
	\item We propose the problem of designing and connectivity checking of implicit user network from the user\mbox{-}item rating data. 
	\item For the \textsc{Implicit User Network Design} Problem, we propose three different approaches, namely, exhaustive search approach, clique addition approach , and matrix multiplication\mbox{-}based approach.
	\item For the \textsc{Implicit User Network Design With Connectivity Checking} Problem, we propose two approaches. First one is the sequential approach: designing and then connectivity checking, and the other one is a concurrent approach: an incremental algorithm, which does the designing and connectivity checking simultaneously.
	\item All the algorithms presented in this paper has been analyzed to understand their time and space requirement.
	\item Proposed algorithms have been implemented with three publicly available user\mbox{-}item rating datasets and an extensive set of experiments have been conducted to understand the efficiency of the algorithms. To investigate the scalability issues of the algorithms they are executed on increasing the input data size. 
\end{itemize}
\par The remaining portion of this article has been arranged in the following way: Section \ref{Sec:Introduction} contains some relevant studies from the literature. Section \ref{Sec:PD} describes some preliminary concepts and define both the problems formally. Proposed methodologies for both the problems have been described in Section \ref{Sec:Proposed}. In Section \ref{Sec:Experimental}, proposed methodologies have been evaluated, and finally, Section \ref{Sec:Conclusion} concludes this study and provides future research  directions.
\section{Related Work} \label{Sec:Introduction}
In this section, some relevant works from the existing literature have been described. This section is broadly divided into two parts. In the first part, we report literature related to the design and analysis of implicit social networks, whereas in the second part we do the same for different applications of the implicit social network.
\subsection{Design and Analysis of Implicit Social Network} 
\cite{gupte2012measuring} proposed an \textit{axiomatic framework}  for measuring the \textit{connectedness} and \textit{tie strength} of an implicit social network. Their methodology is also helpful for inferring implicit relation among a set of people by tie strength.  \cite{li2017influential} proposed a \emph{multi-task low-rank linear influence model} for detecting influential nodes from an implicit social network. \cite{losup2014analyzing} proposed a methodology for designing an implicit social network from real-world data collected from three different game genres which will be beneficial to both players as well as game operators.  \cite{podobnik2015implicit} showed  that implicit social network designed from their on\mbox{-}line behavior actually able to predict hidden relationship among them.  \cite{Taheri} proposed a methodology for extracting implicit social relationship based on rating prediction using the concept \emph{Hellinger Distance}. They have performed social recommendation on this network and their experimental results show that use of implicit user relation in social recommendation methods generate almost identical preferences as explicit trust values. \cite{zhang2014empirical} proposed a methodology to design a implicit brand network from the dataset consisting of historical activities of users on a social media platform. Their experiments answer many interesting research questions about the topology of the brand network, number of users in an influential brand etc. \cite{song2010extraction} proposed a noble methodology for extracting hidden implicit social relationship from messaging cascade. \cite{nauerz2008implicit} proposed a methodology for deriving an implicit social network among the users of a web portal and they have shown that this network can enhance interaction and collaboration in a community. 
\subsection{Applications of Implicit Social Network}
As mentioned in the literature, implicit social networks are found to be useful in designing and improving recommender systems, designing social markets, link prediction in social networks, and so on.  \cite{reafee2016power} showed that the  implicit social network data can be used to improve the \textit{recommendation accuracy} of  \emph{social recommender systems}. \cite{lin2014personalized} proposed a novel Personalized News Recommendation framework using implicit social experts. Their proposed methodology provides better recommendation accuracy specifically for \emph{cold\mbox{-}start users}. \cite{tuarob2015product} have developed a product feature inference model for mining implicit customer preferences within a large scale social media network. \cite{frey2011social} proposed a noble methodology for designing a  \emph{social market} by combining explicit and implicit social relationship. \cite{roth2010suggesting} proposed an interaction\mbox{-}based metric for measuring the affinity of a particular user of the network to other groups. For creating groups, they have used the user's implicit social graph. Their result demonstrates the importance of \emph{implicit social relationship} as well as \emph{affinity\mbox{-}based ranking}. \cite{ma2011learning} proposed a novel \textit{probabilistic factor analysis framework} which incorporates implicit social relationship for recommendation. After that, there are several works on improving \textit{recommendation accuracy} using implicit social relationship \cite{lin2012premise}, \cite{chen2012personalized}.  \cite{tasnadi2015supervised} proposed a methodology for solving \emph{link prediction problem} based on the implicit user information from the network. \cite{alsaleh2011improving} proposed a hybrid social matching system for recommendation using both user's both explicit as well as implicit relationship. Their result shows that the accuracy of the matching process increases if the implicit data is considered.

 To the best of the author's knowledge, there does not exist any study that constructs implicit social network from the user\mbox{-}item rating data. In this paper, we study two related problems in this direction.
\section{Preliminaries and Problem Definition} \label{Sec:PD}
In this section, we present some preliminary concepts related to this study and describe the implicit user network design problem, and implicit user network design with connectivity checking problem formally. In our study, all the graphs are \emph{simple}, \emph{finite}, and \emph{undirected}. A graph is symbolized as $G(V,E)$, where $V(G)$ and $E(G)$ are the set of vertices and edges of the graph, respectively. For any vertex $v$, we denote its set of neighbors $N(v)$ as $N(v)=\{u: (uv) \in E(G)\}$, and cardinality of the neighborhood is known as degree, i.e., $deg(v)=|N(v)|$. A pair of vertices $v_i$ and $v_j$ are adjacent to each other if the edge $(v_{i}v_{j})$ is present in $G$. A graph is said to be \emph{bipartite} if its vertex set can be partitioned into two parts such that no two vertices of the same part are adjacent to each other. A graph is said to be connected if between every pair of vertices there exists a path. If a graph is not connected then it consists of more than one connected components. Readers require more treatment on basic graph theory please refer to \cite{diestel2000graph}. Next, we define the user\mbox{-}item rating data.

\begin{mydef}[User\mbox{-}Item Rating Data]
This is a weighted bipartite graph $\mathcal{G}(U,I,\mathcal{E}, W)$, where $V(\mathcal{G})= U \cup I$, $U=\{u_1, u_2, \ldots, u_{n_1}\}$ are the set of users and $I=\{i_1, i_2, \ldots, i_{n_2}\}$ are the set of items present in the system. $(u_pi_q) \in E(\mathcal{G})$ if and only if user $u_p$ has rated the item $i_q$. $W$ is the edge weight function that assigns each edge to the corresponding rating value, i.e., $W: E(\mathcal{G}) \longrightarrow \mathbb{Z}^{+}$. 
\end{mydef}
In our study, we work with user\mbox{-}item rating datasets, where ratings are binary. We denote the number of users and items present in the system by $n_1$ and $n_2$, respectively. Traditionally, this data can be represented by a bi\mbox{-}adjacency matrix of size $n_1 \times n_2$, where $(m,n)$\mbox{-}th entry is $1$ if the user $u_m$ has rated the item $i_n$ and $0$ otherwise. However, the real\mbox{-}world rating datasets are represented as a collection of tuples of the form $(u_p, i_q,x)$, which means that the user $u_p$ has rated the item $i_q$ with the rating value $x$. As, we are working with binary rating datasets, the $x$ entry is missing. It is easy to observe that for any $u \in U$, $N(u) \subseteq I$, and for any $i \in I$, $N(i) \subseteq U$. For any positive integer $n$, let $[n]$ denotes the set $\{1, 2, \ldots, n\}$. In this paper, as we are dealing with two different graphs\footnote{In rest of the paper, the words `graph' and `network' has been used interchangeably.}, for the ease of clarity, we use the symbol of the graph as subscript for the neighborhood and degree. As an example, for any user $u$, $N_{G}(u)$ denotes the set of other users with which $u$ is directly connected in $G$ and $N_{\mathcal{G}}(u)$ denotes the set of items that the user $u$ has rated. Next, we define the implicit user network.
\begin{mydef}[Implicit User Network]
An implicit user network corresponding to a user\mbox{-}item rating data is basically an undirected, unweighted graph $G(V, E)$, where the vertex set of $G$ is the set of users present in $\mathcal{G}$ and there will be an edge between two users if they have at least one item, which is rated by both of them, i.e., $V(G)=U$, and for all $p,q \in [n_1]$ and $p \neq q$, $(u_pu_q) \in E(G)$ if and only if $N_{\mathcal{G}}(u_p) \cap N_{\mathcal{G}}(u_q) \neq \emptyset$.
\end{mydef} 
Figure \ref{Fig:Example} shows an example of a user\mbox{-}item rating data and its corresponding implicit user network.
\begin{figure}
\centering
\includegraphics[scale=0.8]{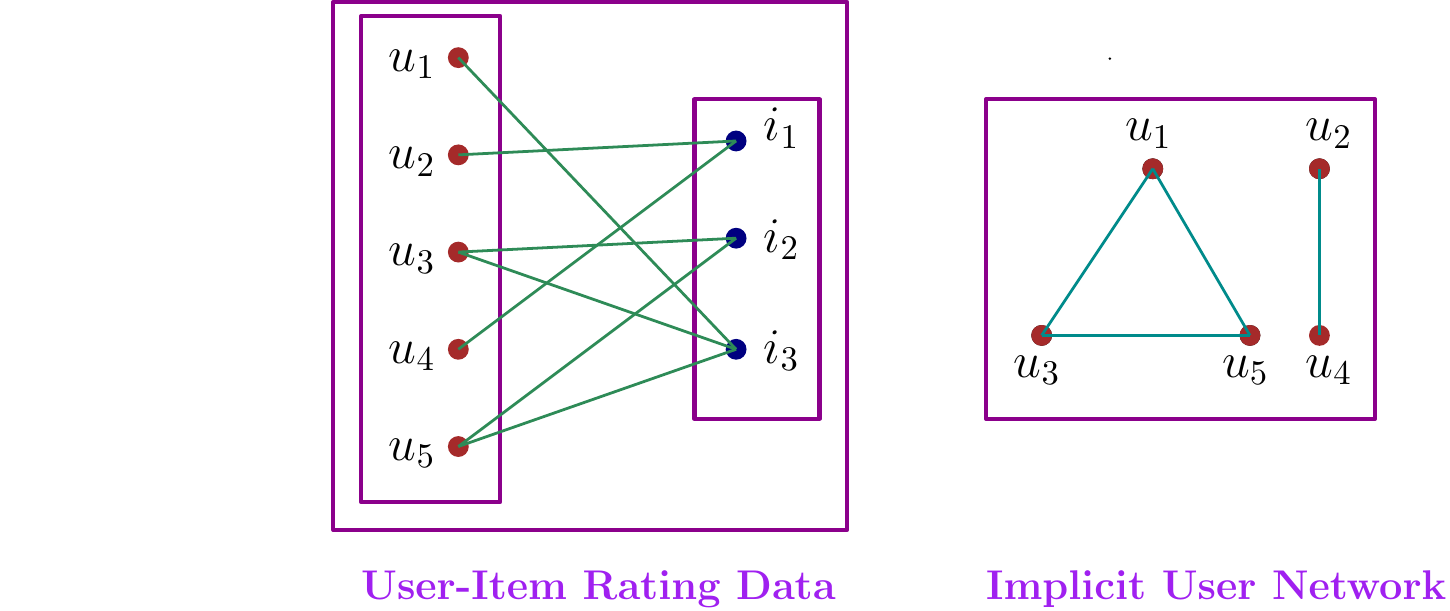} 
\caption{A toy example of user\mbox{-}item rating data and its implicit user network.}
\label{Fig:Example}
\end{figure}
 Next, we define both the problems that we have worked out in this paper.

\begin{tcolorbox}

\underline{\textsc{Implicit User Network Design}} \\
\textbf{Input:} The user\mbox{-}item rating data $\mathcal{G}(U,I,\mathcal{E})$.

\textbf{Problem:} Design the Implicit User Network $G(V,E)$, such that for all $u_p, u_q \in V(G)$, $(u_pu_q) \in E(G)$ if and only if $N_{\mathcal{G}}(u_p) \cap N_{\mathcal{G}}(u_q) \neq \emptyset$.
\end{tcolorbox}
\begin{tcolorbox}

\underline{\textsc{Implicit User Network Design With Connectivity Checking}} \\
\textbf{Input:} The user\mbox{-}item rating data $\mathcal{G}(U,I,\mathcal{E})$.

\textbf{Problem:} Design the Implicit User Network $G(V,E)$, such that for all $u_p, u_q \in V(G)$, $(u_pu_q) \in E(G)$ if and only if $N_{\mathcal{G}}(u_p) \cap N_{\mathcal{G}}(u_q) \neq \emptyset$, and obtain all the connected components $C_1, C_2, \ldots, C_k$ of $G$.

\end{tcolorbox}
Table \ref{tab:Symbols} contains the symbols and notations that are used in this study. Many of them has not been introduced yet. In the next section, the proposed algorithms for both the problems with detailed analysis have been described.

\begin{table}[h]
  \caption{Notations used in this study}
  \label{tab:Symbols}
  \begin{tabular}{||c|c||}
    \toprule
   Notation & Meaning \\
    \midrule
    $\mathcal{G}(U,I, \mathcal{E})$ & User\mbox{-}Item rating data \\
    $U$ & The set of users in $\mathcal{G}$ \\
    $I$ & The set of items in $\mathcal{G}$ \\
    $m$ & Number of edges in $\mathcal{G}$ \\
    $n_1$ & The number of users in $\mathcal{G}$, i.e., $|U|=n_1$\\
    $n_2$ & The number of items in $\mathcal{G}$, i.e., $|I|=n_2$\\
   $G(V,E)$ & The implicit user network among the users in $U$\\
   $V(G)$ & The set of vertices of $G$, i.e., $V(G)=U$\\
   $E(G)$ & The set of edges of $G$ \\
   $m^{'}$ & The number of edges of $G$, i.e., $m^{'}=|E(G)|$ \\
   $\mathcal{B}$ & The Bi\mbox{-}Adjacency matrix of $\mathcal{G}$ \\
   $\mathcal{B}[m,n]$  & $(m,n)$\mbox{-}th entry of $\mathcal{B}$ \\
   $\mathcal{A}$ & Adjacency Matrix of $G$\\
   $\mathcal{B}^{T}$ & Transpose of $\mathcal{B}$ \\
   $|X|$ &  Number of elements in $X$ \\
   $N_{\mathcal{G}}(u)$ &  Neighborhood of the node $u$ in $\mathcal{G}$ \\
   $deg_{\mathcal{G}}(u)$ &  Degree of the node $u$ in $\mathcal{G}$, i.e.,  $deg_{\mathcal{G}}(u)=|\mathcal{N}_{\mathcal{G}}(u)|$\\
   $\Delta_{I}$ &  Maximum degree among the nodes in $I$\\ 
   $\omega$ & Exponent for matrix multiplication\\
   $(u_p)_{\mathcal{G}} \leadsto ^{2}  (u_q)_{\mathcal{G}}$ & $u_p$ is reachable to $u_q$ by a path of length $2$ in $\mathcal{G}$\\
   \bottomrule
\end{tabular}
\end{table}

\section{Proposed Methodologies} \label{Sec:Proposed}
This section is broadly divided into two subsections containing the solution methodologies for the problems with detailed analysis.
\subsection{Solution Methodologies for the Implicit User Network Design Problem} \label{SubSec:31} 
Here, we present three different solution methodologies for the implicit user network design problem.
\subsubsection{Exhaustive Search Approach}
As its name suggests, in this method all the user pairs of $U$  exhaustively checks whether there exists a common item in $I$, which is rated by both the users of the pair. If there exists such an item, then the algorithm puts $1$ in the corresponding entry of the adjacency matrix $\mathcal{A}$ of $G$. Algorithm \ref{Brufo} formally describes the procedure for designing the implicit user network from the user\mbox{-}item rating data.
\begin{algorithm}[h]
	\KwData{User\mbox{-}Item rating data as  Bi-adjacency Matrix ($\mathcal{B}$).}
	\KwResult{ Adjacency Matrix ($\mathcal{A}$) of $G$.}
	$n_1 \leftarrow \text{Number of Rows of }\mathcal{B}$\;
	$n_2 \leftarrow \text{Number of Columns of }\mathcal{B}$\;
	$\text{Create the Matrix }A \text{ of size } n_{1} \times n_{1} \text{and intilialized with }0$\;
	\For{$x=1 $ to $n_1$}{
		\For{$y=x+1 $ to $n_1$}{
			\For{$z=1 $ to $n_2$}{
				\eIf{$B[x][z]==1 \&\& B[y][z]==1$}{
					$\mathcal{A}[x][y] \longleftarrow 1$\;
					$\mathcal{A}[y][x] \longleftarrow 1$\;
					break\;
				}{
				$\mathcal{A}[x][y] \longleftarrow 0$\;
				$\mathcal{A}[y][x] \longleftarrow 0$\;
			}
		}
	}
}
\caption{Exhaustive Search Approach}
\label{Brufo}
\end{algorithm}

Now, we analyze the time and space requirement of Algorithm \ref{Brufo}. As there are $n_1$ users in $\mathcal{G}$, hence the maximum number of possible user pairs could be $\binom{n_1}{2}=\mathcal{O}(n_1^{2})$. Now, for each of the user pair, we need to check whether there exists a common item or not. Hence, for each user pair time requirement is of $\mathcal{O}(n_2)$. So, the total time requirement is of $\mathcal{O}(n_1^{2}n_2)$. Extra space consumed by  Algorithm \ref{Brufo} is to store the adjacency matrix of $G$, which is of $\mathcal{O}(n_1^{2})$. Hence, Theorem \ref{Th:Run_Algo_1} holds.

\begin{mytheo} \label{Th:Run_Algo_1}
Running time and space requirement of Algorithm \ref{Th:Run_Algo_1} is of $\mathcal{O}(n_1^{2}n_2)$ and $\mathcal{O}(n_1^{2})$, respectively.
\end{mytheo}
\subsubsection{Clique Addition Approach}
This methodology works based on the principle stated in Lemma \ref{Lemma:1}.

\begin{mylemma} \label{Lemma:1}
	Let $\mathcal{G}(U,I,\mathcal{E})$ be the user\mbox{-}item rating data, then $\forall i_q \in I$, $N_{\mathcal{G}}(i_q)$ will be a clique in $G$.
\end{mylemma}

\begin{proof}
It has been mentioned previously, for any bipartite graph $\mathcal{G}(U,I,\mathcal{E})$, $\forall i_q \in I$, $N(i_q) \subseteq U$. Now, for any two users $u_x,u_y \in N(i_q)$ for some $q \in [n_2]$, they have always $i_q$ as a common item,and hence, $i_q \in N(u_p) \cap N(u_q)$. This holds for every user pairs of $N(i_q)$. Hence, $\forall i_q \in I$, $N(i_q)$ of $\mathcal{G}$ will be a clique in $G$. 
\end{proof}
As an example, it can be observed from Figure \ref{Fig:Example} that $N_{\mathcal{G}}(i_3)=\{u_1,u_3,u_5\}$ and this is a clique in $G$. Based on the clique addition approach, the implicit user network can be constructed by the following way. Starting with an empty graph where the users in $U$ are the vertices, just add the cliques $N(i_q)$, $\forall i_q \in I$. Algorithm \ref{Algo:Clique} performs this task.

\begin{algorithm}[h]
	\KwData{User\mbox{-}Item rating data as  Bi-adjacency Matrix ($\mathcal{B}$).}
	\KwResult{ Adjacency Matrix ($\mathcal{A}$) of $G$.}
	$n_1 \leftarrow \text{Number of Rows of }\mathcal{B}$\;
	$n_2 \leftarrow \text{Number of Columns of }\mathcal{B}$\;
	$\text{Create the Matrix }\mathcal{A} \text{ of size } n_{1} \times n_{1} \text{and intilialized with }0$\;
	\For{$\text{every item } i \in I$}{
		$N \longleftarrow \emptyset$\;
		\For{$\text{every user } u \in U$}{
			\If{$\mathcal{B}[u][i]==1$}{
				$N \longleftarrow N \cup \{u\}$\;
			}
		}
		\If{$|N| > 1$}{
			\For{$x=1 $ to $|N|$}{
				$a=N[x]$\;
				\For{$y=x+1 $ to $|N|$}{
					$b=N[y]$\;
					$\mathcal{A}[a][b] \longleftarrow 1$\;
					$\mathcal{A}[b][a] \longleftarrow 1$\;
				}
			}
		}
	}
	\caption{Clique Addition Approach}
	\label{Algo:Clique}
\end{algorithm}

Now, we analyze Algorithm \ref{Algo:Clique} to understand its time and space requirements. Let, $\Delta_{I}$ be the maximum degree among the vertices of $I$. Hence, the size of each clique in $G$ due to each item could be as much as $\Delta_{I}$. Starting with an empty graph adding each clique in $G$ requires $\binom{\Delta_{I}}{2}=\mathcal{O}(\Delta_{I}^{2})$ time. As the number of items in the user\mbox{-}item rating data are $n_2$, hence the running time of Algorithm \ref{Algo:Clique} will be $\mathcal{O}(\Delta_{I}^{2} n_2)$. Additional space consumed by Algorithm \ref{Algo:Clique} is due to storing the users that rate the item (refer to Line No. $8$ of Algorithm \ref{Algo:Clique}) which takes $\mathcal{O}(\Delta_{I})$ space and storing the adjacency matrix of $G$, which is of $\mathcal{O}(n_1^{2})$. Hence, the total space requirement of Algorithm \ref{Algo:Clique} is of $\mathcal{O}(n_1^{2}+\Delta_{I}) = \mathcal{O}(n_1^{2})$. Hence, Theorem \ref{Th:Run_Algo_2} holds.

\begin{mytheo} \label{Th:Run_Algo_2}
Running time and space requirement of Algorithm \ref{Th:Run_Algo_2} is of $\mathcal{O}(\Delta_{I}^{2} n_2)$ and $\mathcal{O}(n_1^{2})$, respectively.
\end{mytheo}
Now, it is important to observe that in the worst case $\Delta_{I}$ could be $\mathcal{O}(n_1)$. If for all $i_n \in I$, $deg_{\mathcal{G}}(i_n)=\mathcal{O}(n_1)$, then the running time of Algorithm \ref{Algo:Clique} will be $\mathcal{O}(n_1^{2} n_2)$, which is no better than that of Algorithm \ref{Brufo}. Practically this will be the case when all items are popular items (i.e., rated by many users). However, in reality, rating data are extremely sparse \cite{grvcar2005data}. This means there will be very few items that are rated by many users and the majority of the items are rated by only a few users. In this situation, Algorithm \ref{Algo:Clique} should take less computational time compared to Algorithm \ref{Brufo} and this is exactly what we have observed in our experimentation described in Section \ref{Sec:Experimental}.
\subsubsection{Matrix Multiplication Method}
The intuition behind this method is that if $\mathbb{A}$ be the adjacency matrix of any undirected, unweighted graph then the $(x,y)$\mbox{-}th cell of $\mathbb{A}^{k}$ denotes the length $k$ paths between the vertex $u_x$ and $u_y$ in that graph. Lemma \ref{Lemma:2} describes the fact in this problem context.

\begin{mylemma} \label{Lemma:2}
Let $\mathcal{G}(U,I,\mathcal{E})$ be a user\mbox{-}item rating data and $G(V, E)$ be the designed implicit user network. Now, given any pair of users $u_x$ and $u_y$ of $U$, $(u_xu_y) \in E(G)$ if and only if, they have at least one 2 length path in $\mathcal{G}$. Mathematically, 
	\begin{center}
		$(u_xu_y) \in E(G) \Leftrightarrow (u_x)_{\mathcal{G}} \leadsto ^{2}  (u_y)_{\mathcal{G}}$
	\end{center}
\end{mylemma}

\begin{proof}
As this is an `if and only if' statement, we have to prove both the directions. First, we prove the forward direction. Assume that there exists an edge between $u_p$ and $u_q$ in $G$. This essentially means that there exists minimum one common item in $N_{\mathcal{G}}(u_x) \cap N_{\mathcal{G}}(u_y)$. Without loss of generality, assume that the common item is $i_{n}$, $i_n \in N_{\mathcal{G}}(u_x) \cap N_{\mathcal{G}}(u_y)$. Hence, both the edges $(u_xi_n)$ and $(i_nu_y)$  will be present in $E(\mathcal{G})$. This clearly implies that $u_x$ and $u_y$ are reachable using the path $\langle u_x i_n u_y \rangle$ and the length of this path is two. This necessarily shows that if $(u_xu_y) \in E(G)$ then this implies that $u_x$ and $u_y$ are connected by minimum one path of length $2$ in $\mathcal{G}$. 
\par For the reverse direction, assume that there exists a length $2$ path between $u_x$ and $u_y$ in $\mathcal{G}$. As $\mathcal{G}$ is bipartite, hence there must exist a vertex $i_n$ in $I$ such that $\langle u_x i_n u_y \rangle$ is a path of length $2$. This clearly implies that $i_n \in N_{\mathcal{G}}(u_x) \cap N_{\mathcal{G}}(u_y)$. Hence by definition of implicit social network $(u_xu_y) \in E(G)$. This completes the proof.
\end{proof}

Now, we report another interesting observation in Lemma \ref{Lemma:3}, which relates $\mathcal{B}.\mathcal{B}^{T}$ with $\mathcal{G}$.

\begin{mylemma} \label{Lemma:3}
	Let $\mathcal{G}(U, I, E)$ be the user\mbox{-}item rating data and $\mathcal{B}$ be its bi\mbox{-}adjacency matrix. $(\mathcal{B}\mathcal{B}^{T})[m,m]$ denotes the $(m,m)$\mbox{-}th entry of $\mathcal{B}\mathcal{B}^{T}$. Then $\forall m \in [n_1]$, $(\mathcal{B}\mathcal{B}^{T})[m,m]=deg_{\mathcal{G}}(u_m)$.
\end{mylemma}

\begin{proof}
It has been mentioned before that given a bipartite graph $\mathcal{G}(U, I, E)$ represented as bi\mbox{-}adjacency matrix $\mathcal{B}$ the $(p,q)$-th entry of $\mathcal{B}\mathcal{B}^{T}$ signifies the number of 2 length paths between the vertices $u_p$ and $u_q$ $\forall m,n \in [n_1]$. Hence, in case of $(m,m)$\mbox{-}th entry, this is basically the number of two length paths starting and ending at $u_m$. Now for any vertex $u_m \in V(\mathcal{G})$ if we make traversal of length $2$ from $u_m$ to $u_m$ itself, then one edge incident to $u_m$ will be traversed two times, and the path stars from $u_m$ then goes to some $i_x \in N_{\mathcal{G}}(u_m)$ then again come back to $u_m$. This implies that such 2 length traversal possible is equal to the number of edges incident on $u_m$ and this is same as the degree of $u_m$. Hence, the following relation holds: $\forall m \in [n_1]$, $(\mathcal{B}\mathcal{B}^{T})[m,m]=deg_{\mathcal{G}}(u_m)$. This completes the proof.
\end{proof}

\begin{algorithm}[h]
	\KwData{User\mbox{-}Item rating data as  Bi-adjacency Matrix ($\mathcal{B}$).}
	\KwResult{Adjacency Matrix ($\mathcal{A}$) of $G$.}
	$n_1 \leftarrow \text{Number of Rows of }\mathcal{B}$\;
	$n_2 \leftarrow \text{Number of Columns of }\mathcal{B}$\;
	$\text{Create the matrix } C \text{ of size } n_2 \times n_1 \text{ and initialized with } 0$\;
	$\text{Create the Matrix }\mathcal{A} \text{ of size } n_{1} \times n_{1} \text{and intilialized with }0$\;
	\For{$\text{every user }u \in U$}{
		\For{$\text{every user }i \in I$}{
			$C[i][u]=\mathcal{B}[u][i]$\;
		}
	}
	$A \longleftarrow \text{Multiply the matrices }\mathcal{B} \text{ and } C$\;
	\For{$x=1 $ to $n_1$}{
		\For{$y=1 $ to $n_1$}{
			\If{$x==y$}{
				$\mathcal{A}[x][y] \longleftarrow 0$\;
			}
		}
	}
	\For{$x=1 $ to $n_1$}{
		\For{$y=1 $ to $n_1$}{
			\If{$\mathcal{A}[x][y]>1$}{
				$\mathcal{A}[x][y] \longleftarrow 1$\;
			}
		}
	}
	\caption{Matrix Multiplication\mbox{-}Based Approach}
	\label{Algo:Matmul}
\end{algorithm}
Algorithm \ref{Algo:Matmul} designs the implicit user network based on the matrix multiplication\mbox{-}based approach, whose working principle is as follows. For the given user\mbox{-}item rating data as a bi\mbox{-}adjacency matrix, first, it computes its transpose $\mathcal{B}^{T}$ (Line No. $5$ to $7$). The time requirement for this step is of $\mathcal{O}(n_1n_2)$. Next, it performs the matrix multiplication between $\mathcal{B}$ and $\mathcal{B}^{T}$ (Line No. $8$). Complexity issues of this step is discussed little later. Let, $A$ be the obtained matrix, which is of size $n_1 \times n_1$. Finally, we change the principal diagonal elements of $A$ to $0$, and the other elements which are greater than $1$ to $1$ (Line No. $9$ to $12$, and $13$ to $16$, respectively). Computational time requirement  of this step is of $\mathcal{O}(n_1^{2})$. If, the naive matrix multiplication technique is applied to multiply $\mathcal{B}$ (with  dimension $n_1 \times n_2$) and $\mathcal{B}^{T}$ (with dimension $n_2 \times n_1$), then the computational time requirement of this step will be of $\mathcal{O}(n_1^{2}n_2)$. In that case, the running time of this algorithm will be of $\mathcal{O}(n_1n_2 + n_1^{2}n_2+ n_1^{2})=\mathcal{O}( n_1^{2}n_2)$, which is no better than Algorithm \ref{Brufo}. However, there exist faster rectangular matrix multiplication Algorithms \cite{blaser2013fast}. One of them is due to \cite{le2012faster} and it has been shown that two rectangular matrices of size $n_1 \times \floor{n_1^{k}}$ and $\floor{n_1^{k}} \times n_1$ with $k \leq 0.30298$ can be multiplied with $\mathcal{O}(n_1^{2+ o(1)})$. This $2+o(1)$ is represented as $\omega$ and referred to as the \emph{matrix multiplication exponent}, where $2 \leq \omega \leq 2.374$ \cite{chiantini2018polynomials}. Hence, the running time of Algorithm \ref{Algo:Matmul} is $\mathcal{O}(n_1n_2 + n_1^{\omega}+ n_1^{2})=\mathcal{O}( n_1^{\omega})$. Additional space consumed by  Algorithm \ref{Algo:Matmul} is due to storing the $\mathcal{B}^{T}$, which requires $\mathcal{O}(n_1n_2)$ space, and $A$, which requires $\mathcal{O}(n_1^{2})$ space. Hence, the total space required by Algorithm \ref{Algo:Matmul} is $\mathcal{O}(n_1n_2+n_1^{2})=\mathcal{O}(n_1(n_1+n_2))$. Hence, Theorem \ref{Th:Run_Algo_3} holds.

\begin{mytheo} \label{Th:Run_Algo_3}
Algorithm \ref{Th:Run_Algo_3} can be implemented with $\mathcal{O}(n_1^{\omega})$ time, and $\mathcal{O}(n_1(n_1+n_2))$ time and space requirement respectively, where $\omega$ is the exponent for the rectangular matrix multiplication.
\end{mytheo}
The following example demonstrates the working principle of Algorithm \ref{Algo:Matmul}. If, we multiply the bi\mbox{-}adjacency matrix ($\mathcal{B}$) of $G$ with its transpose ($\mathcal{B}^{T}$) we have the following
\begin{center}
	$\mathcal{B}\mathcal{B}^{T}= 
	\begin{bmatrix}
	0 & 0 & 1 \\
	1 & 0 & 0 \\
	0 & 1 & 1 \\
	1 & 0 & 0 \\
	0 & 1 & 1 
	\end{bmatrix}
	\times 
	\begin{bmatrix}
	0 & 1 & 0 & 1 & 0 \\
	0 & 0 & 1 & 0 & 1 \\
	1 & 0 & 1 & 0 & 1 
	\end{bmatrix}
	$
	\hspace{0.2 cm}
	$\Rightarrow \mathcal{B}\mathcal{B}^{T}= 
	\begin{bmatrix}
	1 & 0 & 1 & 0 & 1 \\
	0 & 1 & 0 & 1 & 0 \\
	1 & 0 & 2 & 0 & 2 \\
	0 & 1 & 0 & 1 & 0 \\
	1 & 0 & 2 & 0 & 2
	\end{bmatrix}
	$ 
\end{center}
As an example, it can be verified from Figure \ref{Fig:Example}, that degree of $u_1$, $u_2$, $u_4$ are $1$, and $u_3$, $u_5$ have degree $2$ in $\mathcal{G}$. Now, putting 0 to all the principal diagonal entries and all other entries that are greater than one to one of $\mathcal{B}\mathcal{B}^{T}$ we have
\begin{center}
	$C= 
	\begin{bmatrix}
	0 & 0 & 1 & 0 & 1 \\
	0 & 0 & 0 & 1 & 0 \\
	1 & 0 & 0 & 0 & 1 \\
	0 & 1 & 0 & 0 & 0 \\
	1 & 0 & 1 & 0 & 0   
	\end{bmatrix}
	$ 
\end{center} 
Now, it can be easily verified that the matrix $C$ is same as the adjacency matrix ($A$) of the implicit user network $G(V, E)$ that has been shown in Figure \ref{Fig:Example}.
\subsection{Implicit User Network Design with Connectivity Checking}
As mentioned previously the connectivity information among the users of the implicit user network is important for different commercial applications by the E\mbox{-}Commerce house which includes viral marketing, computational advertisement, and so on. Here, we address this issue by the following two methods.

\subsubsection{Method 1 (Sequential Approach: Designing and then Connectivity Checking)}
Algorithm \ref{Algo:4} describes the easiest approach for solving the designing and connectivity checking problem. In this approach, first using any one of the three algorithms presented in Section \ref{SubSec:31} the implicit user network is designed and subsequently the breadth first search is run on the designed network to obtain its connected components.

\begin{algorithm}[H]
	\KwData{Bi-adjacency Matrix ($\mathcal{B}$) of $\mathcal{G}$.}
	\KwResult{Adjacency Matrix ($\mathcal{A}$) of $G$ and $C_1, C_2, \ldots, C_k$, where $\forall j \in [k]$, $C_j$ is a connected component.}
	Step 1: Apply any one of Algorithm \ref{Brufo} or \ref{Algo:Clique} or \ref{Algo:Matmul} to design the network. \\
	Step 2: Run \emph{Breadth First Search (BFS)} Algorithm for finding the connected components.
	\caption{Sequential approach for designing and connectivity checking of implicit social network}
	\label{Algo:4}
\end{algorithm}
Before analyzing Algorithm \ref{Algo:4}, we first state and prove the following lemma.
\begin{mylemma} \label{Lemma:4}
Even if the user\mbox{-}item rating data is sparse, the implicit user network may be dense.
\end{mylemma} 
\begin{proof}
Assume that $\mathcal{G}(U,I,\mathcal{E})$ is a user\mbox{-}item rating data with $|U|=n_1$, and $|I|=n_2$. Now, $\mathcal{G}$ is sparse if $E(\mathcal{G})=\mathcal{O}(n_1+n_2)$. It is trivial that for any $i_p \in I$, $1 \leq deg_{\mathcal{G}} (i_p)\leq n_1$. Assume that in $\mathcal{G}$, constant number (say $c$) of items have their degree $\mathcal{O}(n_1)$ and remaining $(n_2-c)$ number of items have degree $\mathcal{O}(1)$. Now, the number of edges of $\mathcal{G}$ will be equal to the sum of the degrees of the items, i.e., $m=\underset{i_p \in I}{\sum}deg_{\mathcal{G}}(i_p)$. Now, the sum of the degrees of the items can be given by the following equation:
\begin{equation}
\underset{i_p \in I}{\sum}deg_{\mathcal{G}}(i_p)=c. \mathcal{O}(n_1) +(n_2-c).\mathcal{O}(1)
\end{equation}
\begin{center}
$\underset{i_p \in I}{\sum}deg_{\mathcal{G}}(i_p)=\mathcal{O}(n_1+n_2)$
\end{center}
This shows that if a few number of items have their degree as $\mathcal{O}(n_1)$ and reaming items have degree $\mathcal{O}(1)$ then it leads to a sparse user\mbox{-}item rating data. Now, pick any item having $\mathcal{O}(n_1)$ neighboring users in $\mathcal{G}$. As per Lemma \ref{Lemma:1}, this item will induce a clique of size $\mathcal{O}(n_1)$ in $G$. A clique of $\mathcal{O}(n_1)$ vertices will have $\mathcal{O}(n_1^{2})$ edges. This means the implicit user network $G$ will also have $\mathcal{O}(n_1^{2})$ edges. This means that the implicit social network is dense. Hence, even sparse user\mbox{-}item rating data may also lead to a dense implicit user network. This completes the proof.
\end{proof}
Now, we analyze Algorithm \ref{Algo:4} for its time and space requirement. Let, $m^{'}$ be the number of edges present in the implicit social network. Hence, performing BFS on $G$ requires $\mathcal{O}(n_1+m^{'})$ time. As shown in Lemma \ref{Lemma:4}, even for sparse user\mbox{-}item rating data, $m^{'}=\mathcal{O}(n_1^{2})$. Hence, the time requirement for performing BFS is of $\mathcal{O}(n_1^{2})$. As simple implementation of BFS requires linear space, hence additional space requirement for performing the BFS is $\mathcal{O}(n_1)$. It is natural that the running time of Algorithm \ref{Algo:4} will depend upon which algorithm is used for designing the implicit user network. If we use Algorithm \ref{Algo:Clique} for designing the network the time and space requirement by Algorithm \ref{Algo:4} will be $\mathcal{O}(\Delta_{I}^{2} n_2 + n_1^{2})$, and $\mathcal{O}(n_1^{2}+n_1)=\mathcal{O}(n_1^{2})$, respectively. Hence, Theorem \ref{Th:Run_Algo_4} holds. 
\begin{mytheo} \label{Th:Run_Algo_4}
Designing and connectivity checking of the implicit user network can be done in $\mathcal{O}(\Delta_{I}^{2} n_2 + n_1^{2})$ time and $\mathcal{O}(n_1^{2})$ space.
\end{mytheo}
However, we can do the designing and connectivity checking of the implicit user network at the same time and it is much beneficial in terms of computational time. We describe this method in the following section.
\subsubsection{Method 2 (Concurrent Approach: Designing and Connectivity Checking Simultaneously)} 
For a given user\mbox{-}item rating data, Algorithm \ref{Algo:Design_and _Connectivity_Checking} performs designing and connectivity checking of the implicit user network simultaneously. As we observe in the experimentation, this method is much more efficient than Algorithm \ref{Algo:4}. Here, we describe the working principle of Algorithm \ref{Algo:Design_and _Connectivity_Checking}. Line $1$ to $8$ are mostly initialization statements where we create the adjacency matrix ($A$) of $G$, a boolean array $Status$ of length $n_2$ and initialized both of them to $0$. $Status(i_p)=1$ means that the clique $N_{\mathcal{G}}(i_p)$ has been added into the implicit user network. Rest part of the algorithm works as follows. If the entire user set has not been  exhausted yet, then start a new component and randomly pick a user from the remaining set of users. Let, the randomly chosen user be $u$. Next step is to find out the neighbor(s) of $u$ in $\mathcal{G}$. After that, for every item in $N_{\mathcal{G}}(u)$, the following steps are  performed. 
\begin{itemize}
\item Pick an item from the list, and if its corresponding entry in $Status$ is $0$ (which means the clique consisting of the vertices of its neighborhood in $\mathcal{G}$ has not been added) then invoke the $Add\_Clique$ function. This function performs the following task. If the neighborhood size is $1$, then it just returns, else for every pair of vertices of the neighborhood, it puts $1$ in the corresponding entries of $\mathcal{A}$.
\item Once the clique is added, the corresponding entry in the $Status$ vector is set to $1$.
\item Those vertices of the clique which are not in the current connected component, are included in it and excluded from the current set of users.
\item Now, take all the neighborhood users of the item and then for each one of these users pick their neighborhood items in $\mathcal{G}$. Check their entry in the $Status$ vector. If it is $0$ then put the item in the list $N$.
\end{itemize}
  These steps are carried out until the list $N$ becomes empty.  This ends the description of Algorithm \ref{Algo:Design_and _Connectivity_Checking}. Next, we report a few important observations, which will help us to argue the correctness, and also the running time of this algorithm. 

\begin{algorithm}[h]
	\KwData{Bi-adjacency Matrix ($\mathcal{B}$) of $\mathcal{G}$.}
	\KwResult{ Adjacency matrix of $G$ and the components $C_1, C_2, \ldots, C_k$ of $V(G)$.}
	$U \leftarrow \left\{u_1, u_2,\dots\ ,u_{n_1}\right\}$ \tcp*{The Set of Users}\
	$I \leftarrow \left\{i_1, i_2,\dots\ ,i_{n_2}\right\}$ \tcp*{The Set of Items}\
	$n_1 \leftarrow \mathcal{B}.no\_ of\_rows()$ \tcp*{Number of Users}\
	$n_2 \leftarrow \mathcal{B}.no\_ of\_columns()$\tcp*{Number of Items}\
	$N \leftarrow \emptyset$\;
	$CreateMatrix(A, n_1, n_1, 0)$ \tcp*{Adjacency Matrix of User Network}\
	$CreateBooleanVector(Status,n_2,0)$\;
	$j \longleftarrow 0$\;
	\While{$ U \neq \phi$}{
	$j \longleftarrow j+1$\;
	$\text{Create the empty list }C_{j}$\;
	$u \longleftarrow \text{Randomly pick a user from } U$\;
	$N \longleftarrow N_{\mathcal{G}}(u)$\;
	\For{$\text{All } i \in N$}{
	\If{$Status(i)=0$}{
	$Add\_Clique (A, N_{\mathcal{G}}(i))$\;
	$Status(i) \longleftarrow 1$\;
	\For{$\text{ All } v \in N_{\mathcal{G}}(i)$}{
	\If{$v \notin C_j$}{
	$C_j \longleftarrow C_j \cup \{v\}$\;
	}
	}
	$U \longleftarrow U \setminus N_{\mathcal{G}}(i)$\;
	\For{$\text{All }v \in N_{\mathcal{G}}(i)$}{
	\For{$\text{All }p \in N_{\mathcal{G}}(v)$}{
	\If{$Status(p)=0$}{
	$N \longleftarrow N \cup \{p\}$\;
	}
	}
	}
	}
	}
	}
	$\text{Function }Add\_Clique(\mathcal{A}, M) \{$\\
	\eIf{$|M| = 1$}{
	$\text{return}$\;
	}
	{
	\For{$\text{Every pair } u_x,u_y \in M$}{
	$\mathcal{A}[x,y] \longleftarrow 1$\;
	$\mathcal{A}[y,x] \longleftarrow 1$\;
	}
	}
	$\}$
\caption{Concurrent approach for designing and connectivity checking problem.}
\label{Algo:Design_and _Connectivity_Checking}
\end{algorithm}

\begin{myobservation} \label{Ob:1}
In Algorithm \ref{Algo:Design_and _Connectivity_Checking}, number of times \texttt{while} loop in Line No. $9$ will execute is same as the  number of connected components of the implicit user network.
\end{myobservation}
\begin{proof}
We prove this statement by analyzing the control flow of the Algorithm \ref{Algo:Design_and _Connectivity_Checking}. Assume that in the first run of the \texttt{while} loop, at Line No. $12$ the user $u_x$ is chosen. After that, an item (say $i_y$) is picked randomly from $N_{\mathcal{G}}(u_x)$, and the clique consisting of the users of the neighborhood of $i_y$, i.e., $N_{\mathcal{G}}(i_y)$ is added. If the nodes of the clique are not in the current component then they are included into it. Also, all the neighbor items of the users in $N_{\mathcal{G}}(i_y)$ are added to the list $N$. Now, for any $i_x \in N$, it is important to observe that the $N_{\mathcal{G}}(i_x) \cap N_{\mathcal{G}}(i_y) \neq \emptyset$. Hence, $N_{\mathcal{G}}(i_x)$ and $N_{\mathcal{G}}(i_y)$ are connected. Now, applying this argument iteratively, the subgraph induced by the cliques corresponding to the items in $N$ will be connected. The list $N$ becomes empty when the connected component that was currently built is finished. So, once a user is chosen randomly from $U$ it first finishes the construction of the entire connected component in which the randomly chosen user belongs and next the algorithm chooses another user from the remaining set of users uniformly at random to construct another connected component of the implicit user network. Hence, the number of times user will be chosen is the same as the number of connected components. This implies that the number of times the \texttt{while} loop executes will be the same as the number of connected components of the implicit social network. This proves the statement. 
\end{proof}
\begin{myobservation} \label{Ob:2}
The $Add\_Clique$ function of Algorithm \ref{Algo:Design_and _Connectivity_Checking} will be invoked just once for every $i_p \in I$. 
\end{myobservation}
\begin{proof}
In Observation \ref{Ob:1}, it has already been shown that once a user is chosen randomly at Line No. $12$, the entire connected component is built without picking any further user randomly. So, it is sufficient to show that even in one single run of the \texttt{while} loop, for all the items $i_x$ linked with the users of the connected component which is currently being built, the $Add\_Clique$ function is called only once. As soon as the clique corresponding to the item is added, its flag in the $Status$ array is set to $1$. Also, in Line No. $25$ an item with its $Status$ flag $1$ has not been added into the list $N$. Hence, for every item the $Add\_Clique$ function is invoked just once. This completes the proof.
\end{proof}
Now, we analyze Algorithm \ref{Algo:Design_and _Connectivity_Checking} to understand its time and space requirements. From Line No. $1$ to $8$ all are initialization statement, and hence takes $\mathcal{O}(1)$ time. It has been shown in Observation \ref{Ob:1} that the number of times the \texttt{while} loop of Line No. $9$ will run is the same as the number of connected components the implicit social network has. Theoretically, the number of connected components of a $n$ vertex network is of $\mathcal{O}(n)$. However, from practical evidence,  every social network contains a giant component that contains a significant fraction of nodes. As an example, in the Twitter follow graph \cite{myers2014information}, the largest connected component contains $92.9 \%$ of the active users. Hence, in practice the number of connected components is constant. In our analysis also, we assume that the number of connected components of the implicit social network will be constant. In turn, it implies that the \texttt{while} loop will also run for $\mathcal{O}(1)$ time. Inside the loop execution of  the statements from Line No. $10$ to $12$ requires $\mathcal{O}(1)$ time. Now, for the randomly chosen user (say, $u$) computing $N_{\mathcal{G}}(u)$ requires $\mathcal{O}(n_2)$ time. Hence, the \texttt{for} loop at Line No. $14$ will execute $\mathcal{O}(n_2)$ times. It has been shown in Observation \ref{Ob:2} that for every item $i_x \in I$, the $Add\_Clique(.)$ function will be invoked just once. Now, the running time of this function will depend on the size of the clique. Now, assume that the maximum degree among all the items is of $\mathcal{O}(\Delta_{I})$. Hence, the running time of the $Add\_Clique(.)$ function will be of $\mathcal{O}(\Delta_{I}^{2})$. After that, setting the flag to true of the item for which the clique has been added in the implicit user network at Line No. $17$ will require $\mathcal{O}(1)$ time. For any $i_x \in I$, $N(i_x)$ could be of at most $\mathcal{O}(n_1)$. Hence, the number of times the \texttt{for} loop in Line No. $18$ will run is of $\mathcal{O}(n_1)$. Now, the size of a component could be as big as $\mathcal{O}(n_1)$. Hence, the condition checking of the \texttt{if} statement at Line No. $19$ will require $\mathcal{O}(n_1)$ time. Hence, the running time of the \texttt{for} loop from Line No. $18$ to $20$ requires $\mathcal{O}(n_1^{2})$ time. Now, performing the `set minus' operation at Line No. $21$ requires $\mathcal{O}(n_1^{2})$ time. It is quite easy to observe that the \texttt{for} loops at Line No. $22$ and $23$ can execute at most $\mathcal{O}(n_1)$ and $\mathcal{O}(n_2)$ times, respectively. After that, condition checking for the \texttt{if} statement at Line No. $24$ and  adding the item in the List $N$ at Line No. $25$ requires $\mathcal{O}(1)$ time. Now, we unwrap the time requirement of Algorithm \ref{Algo:Design_and _Connectivity_Checking} from bottom to top. Time requirement from Line No. $22$ to $25$ requires $\mathcal{O}(n_1n_2)$ time. Hence, the time requirement from Line No. $21$ to $25$ is of $\mathcal{O}(n_1^{2}+n_1n_2)$. This implies that the time requirement from Line No. $18$ to $25$ is of $\mathcal{O}(n_1^{2}+ n_1^{2} +n_1n_2)= \mathcal{O}(n_1^{2} +n_1n_2)$. Time requirement from Line No. $14$ to $25$ requires $\mathcal{O}(n_1(\Delta_{I}^{2}+ n_{1}^{2}+n_1n_2))=\mathcal{O}(n_1\Delta_{I}^{2}+n_{1}^{3}+n_{1}^{2}n_{2})$. Now, as $\Delta_{I}=\mathcal{O}(n_1)$, hence $\mathcal{O}(n_1\Delta_{I}^{2}+n_{1}^{3}+n_{1}^{2}n_{2})=\mathcal{O}(n_{1}^{3}+n_{1}^{2}n_{2})$. As, the \texttt{while} loop runs for a constant time, hence the time requirement from Line No. $9$ to $25$ is of $\mathcal{O}(n_2+n_{1}^{3}+n_{1}^{2}n_{2})=\mathcal{O}(n_{1}^{3}+n_{1}^{2}n_{2})$. As other statements of Algorithm \ref{Algo:Design_and _Connectivity_Checking} requires $\mathcal{O}(1)$ time, hence the time requirement of this algorithm is of $\mathcal{O}(n_{1}^{3}+n_{1}^{2}n_{2})$. Additional space consumed by Algorithm \ref{Algo:Design_and _Connectivity_Checking} is due to storing $A$, which consumes $\mathcal{O}(n_1^{2})$ space; $Status$, $N$ which requires $\mathcal{O}(n_2)$ space, and storing all the components together requires $\mathcal{O}(n_1)$ space. Hence, the total space requirement of Algorithm \ref{Algo:Design_and _Connectivity_Checking} is of $\mathcal{O}(n_1^{2}+n_2+n_1)=\mathcal{O}(n_1^{2}+n_2)$. Hence, the following theorem holds.
\begin{mytheo}
The time and space requirement of Algorithm \ref{Algo:Design_and _Connectivity_Checking} is of $\mathcal{O}(n_{1}^{3}+n_{1}^{2}n_{2})$ and $\mathcal{O}(n_{1}^{2}+n_{2})$, respectively.
\end{mytheo}
\section{Experimental Evaluation} \label{Sec:Experimental}
In this section, we perform an extensive set of experiments for evaluating the proposed methodologies. Initially, we start with a brief  description of the datasets.
\subsection{Description of the Datasets}
In our experiments, we have used the following user\mbox{-}item rating dataset. All of them are collected from \textit{ Koblenz Network Collection (KONECT)} \footnote{ \url{http://konect.uni-koblenz.de/}}.

\begin{itemize}
\item \textbf{Filmtrust} \cite{guo2013novel}: This is a bipartite rating network between users and movies. An undirected edge between user and item denotes the user has rated the item. Edge weight represents rating value based on a particular rating scale. 
\item \textbf{MovieLens} \footnote{\url{http://konect.uni-koblenz.de/networks/movielens-1m}}: This user\mbox{-}item rating network contains one million movie ratings from \url{http://movielens.umn.edu/}. Left nodes are users and right nodes are movies. An edge between a user and a movie shows that the user has rated the movie.
\item \textbf{Epinions} \footnote{\url{http://konect.uni-koblenz.de/networks/epinions-rating}} \cite{massa2005controversial}: This is the bipartite rating network of Epinions, an on\mbox{-}line product rating site. Each edge connects a user with a product and represents a rating as edge weight. 
\end{itemize} 
As our study is concerned with binary ratings, hence, for all the datasets, we do re\mbox{-}scaling of the rating values in $0-1$ scale. Table \ref{Tab:1} contains the basic statistics of the datasets. As mentioned previously, from the density values presented in Table \ref{Tab:1} one can convince himself that the rating datasets are extremely sparse.

\begin{table}[h]
\caption{Basic Statistics of the Datasets}
\label{Tab:1}       
\begin{tabular}{lllll}
\hline\noalign{\smallskip}
Dataset Name & \# Users & \# Items & \# Ratings & Density \\
\noalign{\smallskip}\hline\noalign{\smallskip}
Filmtrust & 1508 & 2071 & 35497 & 0.011366 \\
MovieLens & 9746 & 6040 & 1000209 & 0.01699\\
Epinions & 40163 & 139738 & 664824 & 0.000118\\
\noalign{\smallskip}\hline
\end{tabular}
\end{table}

\subsection{Experimental Setup}
As there is no prior work on the designing implicit user network from the user\mbox{-}item rating data, we can not compare the performance of the methods with any existing methods. Instead, we do a comparative study among the proposed methodologies itself. All the proposed algorithms have been implemented on Python 2.7 with NetworkX 1.9.1 Package. All the experiments have been carried out using a 5 node high performance computing cluster each of them has 32 cores and 64 GB of RAM.
\par As our goal is to make a comparative study regarding computational time and scalability of the proposed algorithms, for each of the datasets, we start with $\frac{1}{5}$ number of ratings of the original dataset from the top, and subsequently add $\frac{1}{5}$ more, and continued until the whole dataset is exhausted. As the exhaustive search method is taking huge computational time, we do not report results for this method on larger datasets (i.e., other than the `Filmtrust'). 
\subsection{Experimental Results with Observation}
Here, we describe our obtained results and list out the key observations. We start with reporting the results for the implicit user network design problem. 

\begin{figure}[ht]
\begin{subfigure}{.5\textwidth}
  \centering
  \includegraphics[width=.8\linewidth]{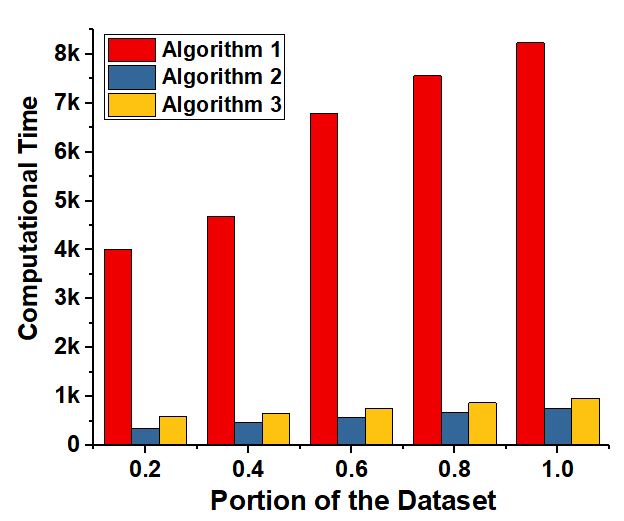}  
  \caption{Flimtrust Dataset}
  \label{fig:sub-first}
\end{subfigure}
\begin{subfigure}{.5\textwidth}
  \centering
  \includegraphics[width=.8\linewidth]{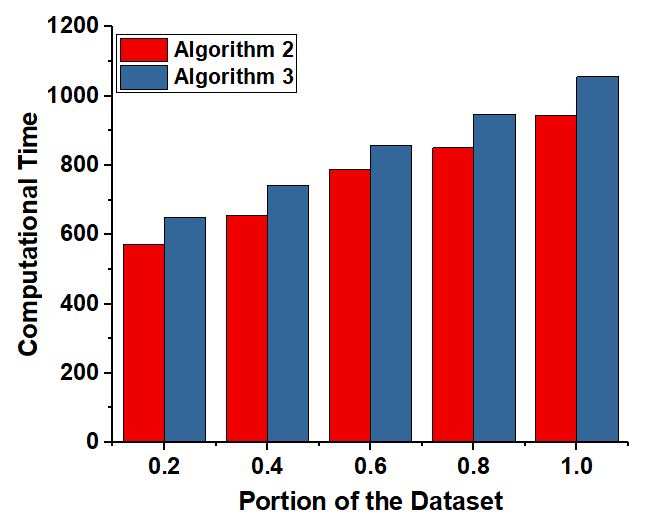}  
  \caption{Movielens Dataset}
  \label{fig:sub-second}
\end{subfigure}
\vspace{0.5 cm}
\begin{center}
\includegraphics[width=.42\linewidth]{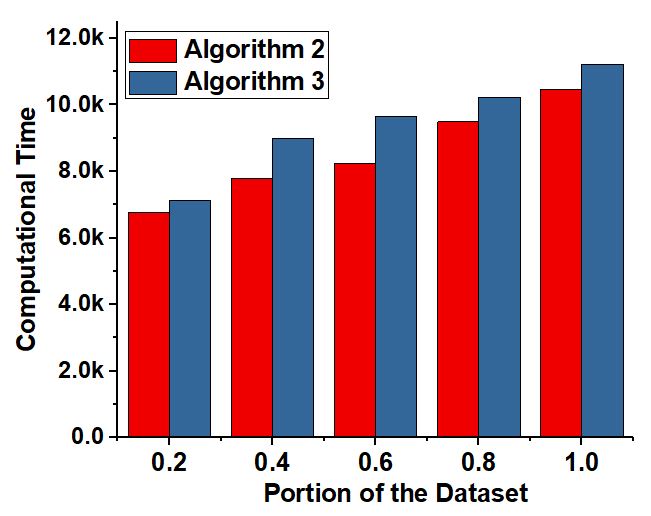} \\
(c) Epinions Dataset 
\end{center}
\caption{Portion of the dataset vs. computational time plots for the implicit user network design problem on different datasets. Here, Algorithm \ref{Brufo}, \ref{Algo:Clique}, and \ref{Algo:Matmul} denotes the `Exhaustive Search Approach', `Clique Addition Approach', and `Matrix Multiplication\mbox{-}based Approach', respectively.}
\label{Fig:2}
\end{figure}

\par Figure \ref{Fig:2} shows the portion of the dataset versus computational time plots for the implicit user network design problem. From Figure \ref{Fig:2}a, it can be clearly observed that the time requirement for the `exhaustive search method' is very very high compared to the both `clique addition approach' and `matrix multiplication\mbox{-}based' approach. As an example, in case of `Filmtrust' dataset, when only top $\frac{1}{5}$ of the total number of ratings have been used, the computational time requirement for the implicit user network design for the exhaustive search approach is $4000$ Secs. However, the same for the `clique addition approach' and the `matrix multiplication\mbox{-}based approach' are $356$ Secs and $596$ Secs, respectively. The key observations are as follows:
\begin{itemize}
\item Among the proposed approaches, the `exhaustive search approach' takes the maximum computational time, as for every pair of users this method searches the entire item set to check for the existence of a common item.
\item Among the remaining two methods, from the experiments, it has been observed that the `clique addition approach' is much faster than the other method. As an example, when the whole `Epinion' dataset has been used, the time requirement to construct the implicit social network by the `clique addition approach' and `matrix multiplication\mbox{-}based approach' are $10456$ seconds and $11226$ seconds, respectively.
\end{itemize}
So, it can be concluded that the `clique addition approach' is the superior one and should be used to construct the implicit user network, particularly when the rating dataset is sparse. Next, we proceed to report the experimental results related to the `implicit user network design with connectivity checking' problem.

\begin{figure}[ht]
\begin{subfigure}{.5\textwidth}
  \centering
  \includegraphics[width=.8\linewidth]{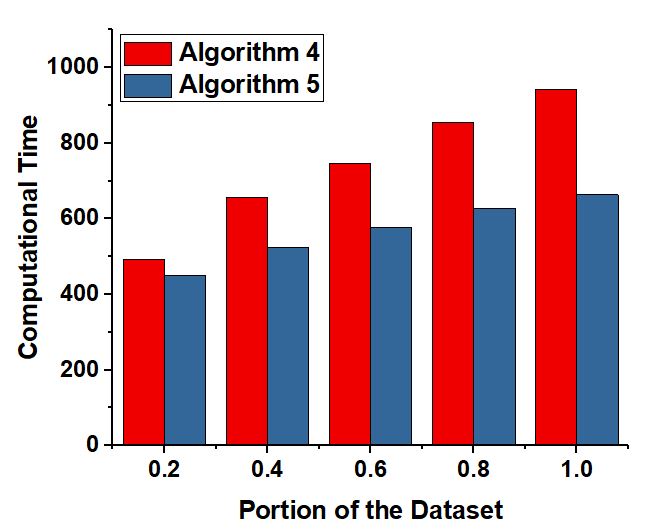}  
  \caption{Flimtrust Dataset}
  \label{fig:sub-first}
\end{subfigure}
\begin{subfigure}{.5\textwidth}
  \centering
  \includegraphics[width=.8\linewidth]{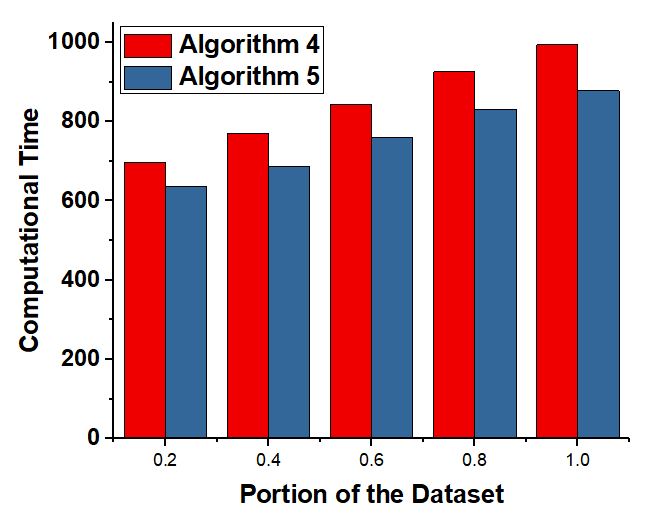}  
  \caption{Movielens Dataset}
  \label{fig:sub-second}
\end{subfigure}
\vspace{0.5 cm}
\begin{center}
\includegraphics[width=.42\linewidth]{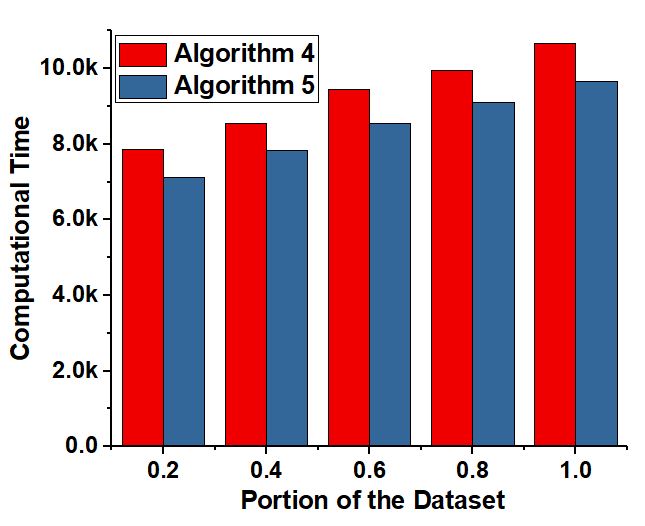} \\
(c) Epinions Dataset 
\end{center}
\caption{Portion of the dataset vs. computational time plots for the implicit user network design and connectivity checking problem on different datasets. Here, Algorithm \ref{Algo:4}, and \ref{Algo:Design_and _Connectivity_Checking} refers to the Sequential Approach: First Design and then Connectivity Checking, and Concurrent Approach: Design and Connectivity Checking simultaneously.}
\label{Fig:3}
\end{figure}

\par Figure \ref{Fig:3} shows the portion of the ratings used vs. computational time plots for all the three datasets. From this figure, it can be observed that across all the datasets, the concurrent approach: designing and connectivity checking together (i.e., Algorithm \ref{Algo:Design_and _Connectivity_Checking}) is much more efficient compared to the sequential approach: design and then connectivity checking method (Algorithm \ref{Algo:4}). As an example, when the entire `MovieLens' dataset has been used, the  computational time requirement for Algorithm \ref{Algo:Design_and _Connectivity_Checking} and \ref{Algo:4} is $878$ and $994$, respectively. The reason behind this is as follows: at the time of designing, the edges of the implicit user network are traversed, and also, at the time of connectivity checking the same traversal is happening once more which is redundant. As in the second method, we are cleverly maintaining the connectivity information during the designing itself, this saves the computational time.

\section{Conclusion} \label{Sec:Conclusion}
In this paper, we have introduced two related problems regarding the designing and connectivity checking of the implicit user network from the user\mbox{-}item rating data. For the implicit user network design problem, we have proposed three different approaches, namely exhaustive search approach, clique addition approach, and matrix multiplication\mbox{-}based approach. For the implicit user network design with connectivity checking problem, we have proposed two different approaches. The first one is the sequential approach: designing and then connectivity checking, and the other one is a concurrent approach: an incremental algorithm, which does the designing and connectivity checking simultaneously. Experimentation with real\mbox{-}world user\mbox{-}item rating datasets show that for the first problem the `clique addition approach' performs better than the rest of the two approaches since the datasets are extremely sparse. For the second problem, it is observed that the concurrent approach takes less computational time. Now, it will be an interesting future work to use the implicit social network information for social recommendation, seed set selection for viral marketing and study its performance.
%


%
%

\bibliographystyle{spbasic}      
\bibliography{sigproc}   


\end{document}